\documentclass[12pt,journal]{IEEEtran}

\usepackage[hmargin=1.2in,vmargin=1.2in]{geometry}
\usepackage{amsmath,amsthm,amssymb,amsfonts,verbatim}
\usepackage{graphicx}
\usepackage[colorlinks, linkcolor=blue,  anchorcolor=blue,
            citecolor=blue
            ]{hyperref}
\usepackage[hmargin=1.2in,vmargin=1.2in]{geometry}
\usepackage{booktabs}
\usepackage{bm}

\title{New families of non-Reed-Solomon MDS codes}
\author{Lingfei Jin, Liming Ma, Chaoping Xing and Haiyan Zhou
\thanks{L. Jin is with the Shanghai Key Laboratory of Intelligent Information Processing, School of Computer Science, Fudan University, Shanghai 200433, China, and with the State Key Laboratory of Cryptology, P. O. Box, 5159, Beijing 100878, China (emails: lfjin@fudan.edu.cn).}
\thanks{L. Ma is with the School of Mathematical Sciences, University of Science and Technology of China, Hefei 230026, China (e-mail: lmma20@ustc.edu.cn).}
\thanks{C. Xing is with School of Electronic Information and Electrical Engineering, Shanghai Jiao Tong University, Shanghai 200240, China (email: xingcp@sjtu.edu.cn).}
\thanks{H. Zhou is with School of Mathematical Sciences, Nanjing Normal University, Nanjing 210023, China (email: 05336@njnu.edu.cn).}

}

\newtheorem{lemma}{Lemma}[section]
\newtheorem{theorem}[lemma]{Theorem}
\newtheorem{maintheorem}[lemma]{Main Theorem}
\newtheorem{cor}[lemma]{Corollary}

\newtheorem{prop}[lemma]{Proposition}
\newtheorem{ex}[lemma]{Example}
\newtheorem{rem}[lemma]{Remark}
\newtheorem{defn}{Definition}

\theoremstyle{remark}
\newtheorem{rmk}{Remark}

\renewcommand{\epsilon}{\varepsilon}
\renewcommand{\le}{\leqslant}
\renewcommand{\ge}{\geqslant}

\newcommand{\vnote}[1]{}

\def\F{\mathbb{F}}

\newcommand{\Ga}{\alpha}
\newcommand{\Gb}{\beta}

\def \bc {{\mathbf c}}

\def \bx {{\mathbf x}}

\def \by {{\mathbf y}}

\def \bv {{\mathbf v}}

\def\span{{\rm span}}
\def\bGa{\boldsymbol{\alpha}}

\newcommand \mcal \mathcal

\onecolumn

\begin{document}
\maketitle

\begin{abstract}
MDS codes have garnered significant attention due to their wide applications in practice. To date, most known MDS codes are equivalent to Reed-Solomon codes.
The construction of non-Reed-Solomon (non-RS) type MDS codes has emerged as an intriguing and important problem in both coding theory and finite geometry.
Although some constructions of non-RS type MDS codes have been presented in the literature, the parameters of these MDS codes remain subject to strict constraints.
In this paper, we introduce a  general framework of constructing $[n,k]$ MDS codes using the idea of selecting a suitable set of evaluation polynomials and a set of evaluation points such that all nonzero polynomials have at most $k-1$ zeros in the evaluation set. Moreover, these MDS codes can be proved to be non-Reed-Solomon  by computing their Schur squares. Furthermore, several explicit constructions of  non-RS MDS codes  are given by converting to combinatorial problems. As a result, new families of non-RS MDS codes with much more flexible lengths can be obtained and most of them are not covered by the known results.
\end{abstract}

%{\bf Index Terms---MDS codes, Reed-Solomon codes, Schur square}

\section{Introduction}
An $[n,k,d]$-linear code $C$ over the finite field $\F_q$ is a $k$-dimensional vector subspace of $\F_q^n$ with the minimum Hamming distance $d$.
The well-known Singleton bound states that $d\le n-k+1$.
If the minimum distance of a code $C$ attains the Singleton bound, then $C$ is called a maximum distance separable (MDS) code. Since the birth of coding theory, MDS codes have received great attention for their maximal error correcting capability.
As a result, various properties of MDS codes have been well investigated, such as the existence of MDS codes \cite{Dau14}, classification  of MDS codes \cite{Kok15}, balanced MDS codes \cite{Dau}, the lowest density MDS codes \cite{Blaum} and  non-Reed-Solomon MDS codes \cite{RL}.
In addition, MDS codes are closely connected to combinatorial designs and finite geometry \cite{Mac}.

Among them, the existence of MDS codes have achieved much attention.
There is a famous conjecture on MDS codes which says that the length of an MDS code over $\F_q$ is at most $q+2$, more precisely, $n\le q+1$ except for some exceptional cases \cite{Se55}.
This conjecture was proved for MDS codes over prime fields by Ball \cite{B12}.
Until now, the best known MDS codes are the so-called  Reed-Solomon (RS) codes and their generalizations  which  have found wide applications in practice.
We know that there exist $[n\le q,k,n-k+1]$ Reed-Solomon (RS) codes over $\F_q$ for $1\le k\le q$.
In fact, the lengths of extended generalized Reed-Solomon (GRS) codes can go up to $q+1$.
Moreover, the dual of a GRS code is also a GRS code.
In the literature, properties of GRS codes have been extensively studied for their wide applications such as in cryptography and distributed storage systems.

In coding theory, MDS codes which are not equivalent to Reed-Solomon codes are called non-Reed-Solomon type MDS codes.
Since most of the known MDS codes are equivalent to RS codes, constructing non-RS type MDS codes have become an interesting research topic.
In 1989, Roth and Lempel gave a class of non-RS MDS codes via generator matrices by employing the subsets of the finite field $\F_q$ with some special properties \cite{RL}. They showed that the value of $n$ is at least $q/2+2$ for even $q$ and $3\le k\le q/2-1$, while $n\ge k+3$ for even $q$ and $q/2\le k\le q-1$.
Inspired by the work \cite{Sheekey} of Sheekey on twisted Gabidulin codes,
Beelen {\it et al.} introduced a representative work given in \cite{Beelen17, Beelen22} which proposed a new class of MDS codes called twisted RS codes.
The idea of constructing  twisted RS codes is using spaces of polynomials that may contain elements of degree larger than $k$ for evaluation.
However, twisted RS codes are not MDS in general.
The authors presented two explicit subclasses of twisted RS codes that are MDS by using the properties of multiplicative subgroups and additive subgroups respectively. Thus, they can obtain MDS codes with length $(n-1)|(q-1)$ and $(n-1)|q$ that are non-RS type.
%such codes  $N(k,q)^{'}$ are about $q/2$ for any field of size $q$ and  most of the twisted codes are not equivalent to an RS code
However, it seems that there are strong restrictions on the  lengths of these MDS twisted Reed-Solomon codes.
% we show that our construction yields several families of MDS codes
In \cite{Beelen18}, a generalization of twisted RS codes are given by using the idea of multi-twists which includes the codes given in \cite{Beelen17} as a special case.
The corresponding dual codes are discussed as well.
Following the work of Beelen {\it et al.}, many work have been done on exploring the structure and properties of twisted RS codes, such as self-dual MDS codes \cite{Huang21,Sui22}, LCD MDS codes \cite{Liu21, Wu21} and near MDS codes \cite{Sui}.

Recently, Chen \cite{Chen24} constructed some new non-RS MDS codes from algebraic curves. Particularly, the author obtained some non-RS MDS codes with length $O(q^{1/4})$ from elliptic curve codes. It was further showed in \cite{Chen24} that the length of non-RS MDS codes can achieve $O(q^{1/k})$ from algebraic curves with higher genus under certain conditions.
Li  {\it et al.}  designed some new cyclic MDS codes for certain parameters by determining the solutions of the system of polynomial equations
\cite{LC24}. It was further shown that those codes are non-RS type. The MDS codes obtained have length $n|(q-1)$  where $q=p^r$ for certain $p$, which also has strong restrictions. Wu  {\it et al.} studied the extended codes of Roth-Lempel codes \cite{Wu24}. They provided a sufficient and necessary condition of such codes being non-RS MDS codes, however, they only provided explicit non-RS MDS codes over small finite fields.
Therefore, the existence and explicit constructions  of non-RS type MDS codes for all possible parameters remain an open problem.

\subsection{Our approaches and results}

It is known that any MDS code of length $n$ and dimension $k$ is equivalent to a RS code for $k<3$ and $n-k<3$ \cite{Beelen17}. Thus, we only consider the construction of non-RS MDS codes for  $3\le k\le n-3$. As the dual code of a RS code is again a RS code, hence we can further assume that $3\le k\le n/2$ in the following context.

Motivated by the constructions in \cite{Beelen17,LC24}, we provide a general framework  of constructing non-RS type MDS codes  by polynomial evaluations where the polynomial space is spanned by some monomials. More specifically, we choose pairwise distinct polynomial evaluation points $\Ga_1,\Ga_2,\dots,\Ga_n$ from $\F_q$ and some non-negative integers $i_1<i_2<\dots<i_k$ and consider the evaluation code
\[C(\bGa,I)=\{(f(\Ga_1),f(\Ga_2),\dots,f(\Ga_n)):\; f\in \span_{\F_q}\{x^{i_1},x^{i_2},\dots,x^{i_k}\}\},\]
where $\bGa=\{\Ga_1,\Ga_2,\dots,\Ga_n\}$ and $I=\{i_1,i_2,\dots,i_k\}$. Without loss of generality, we may assume that $i_1=0$ (otherwise, we get an equivalent code).

To show that $C(\bGa,I)$ is a non-RS MDS code, we have to show that
\begin{itemize}
\item[(i)] $C(\bGa,I)$ is an $[n,k,n-k+1]$-MDS code;
\item[(ii)] $C(\bGa,I)$ is not equivalent to an RS code.
\end{itemize}
The length and dimension of $C(\bGa,I)$ is clear.
To show the minimum distance of $C(\bGa,I)$ is $n-k+1$, it is sufficient to show that any nonzero polynomial $f(x)\in V_I:=\span_{\F_q}\{x^{i_1},x^{i_2},\dots,x^{i_k}\}$ has at most $k-1$ roots among $\Ga_1,\dots,\Ga_n$. This is equivalent to the fact that
\[ \deg (\gcd(f(x),\prod_{i=1}^n(x-\alpha_i)))\le k-1.\]
For $k\le n/2$, to show that $C(\bGa,I)$  is not equivalent to an RS code, it is sufficient to show that the Schur square $C^{\star2}(\bGa,I)$
has dimension at least $2k$ (we refer to  Subsection \ref{sub:2.3}  for the definition of Schur square).
This is because the dimension of the Schur square of an $[n,k]$ RS code with $k\le n/2$ has dimension $2k-1$ (see Subsection \ref{sub:2.3}).

One of the main objectives of this paper is to show that $C(\bGa,I)$ are non-RS MDS codes for any non-arithmetic progression sets $I$ (we refer to Subsection \ref{sub:2.2} for the definition of arithmetic progression sets). More precisely speaking, given $3\le k\le n/2$ and sufficiently large $q$, then for any  non-arithmetic progression set  $I$ of size $k$, there exists an evaluation set $\bGa$ of size $n$ such that $C(\bGa,I)$ is a non-RS MDS code. Note that $C(\bGa,I)$ is a non-RS MDS code if and only if the dual $C^\perp(\bGa,I)$ is a non-RS MDS code. Thus,  we may assume that  $3\le k\le n/2$.

For arbitrary  non-arithmetic progression set  $I$, we have the following existence result for non-RS MDS codes.
\begin{maintheorem}\label{main:1}
Let $q$ be a prime power and let $n$ and $k$ be positive integers with $3\le k\le n/2$.
If $\binom{q}{n}>\frac{q^k-1}{q-1}  \binom{m_I}{k}\binom{q-k}{n-k}$ with $m_I=\max\{i:\; i\in I\}$ for any  non-arithmetic progression set  $I\subseteq \mathbb{N}$ of size $k$, then there exists an evaluation set $\bGa$ of size $n$ such that $C(\bGa,I)$ is a non-RS MDS code.
\end{maintheorem}

The Main Theorem \ref{main:1} shows that for a non-arithmetic progression set  $I$ of size $k$ with $3\le k\le n-3$ and $n=O\left(\frac{q^{\frac1k}}{m_I}\right)$ (see Remark \ref{rem:3.5} for the detailed computation), there always exists an $[n,k]$ non-RS MDS code. In particular, we may take $I=\{0,1,\dots,i-1,i+1,i+2,\dots,k\}$ with $1\le i\le k-1$, then $I$ is a  non-arithmetic progression set  and we have $m_I=k$. Thus, there exists an $[n,k]$-non RS MDS code as long as $3\le k\le n/2$ and $n=O\left(q^{\frac1k}\right)$.

\begin{rem}{Recently, non-RS MDS codes were constructed with length $n=O\left(q^{\frac1k}\right)$ using algebraic curves  in \cite{Chen24}. Although our result does not improve upon \cite{Chen24}, the idea behind our construction is simpler and the required condition can be easily verified.}
\end{rem}

 The Main Theorem \ref{main:1} demonstrates  the existence of non-RS MDS codes for an arbitrary set $I$. Another main objective of this paper is to provide explicit constructions of non-RS MDS codes for specific choices of the set $I$.
In particular, we take
 $I=\{0,1,\dots,i-1,i+1,i+2,\dots,k\}$ with $1\le i\le k-1$. We split this case into two subcases, noting that the case where $i=k-1$ yields a more interesting result.

 \begin{maintheorem}\label{main:2}
Let $I=\{0,1,\dots,k-2,k\}$ with $3\le k\le n/2$. Then there exists a $q$-ary  non-RS MDS code $C(\bGa,I)$ which can be explicitly constructed, if one of the following conditions is satisfied.
%Then the evaluation set $\bGa$ of size $n$ can be explicitly constructed such that $C(\bGa,I)$ is a $q$-ary  non-RS MDS code,
\begin{itemize}
\item[{\rm (i)}] $2k\le n<\frac{q}k+\frac{k+1}2$ for a prime $q$ (see Corollary \ref{cor:4.4}); or
\item[{\rm (ii)}] $2k\le n\le  \max\{1, [\frac{ p}{k}]\}\cdot \frac{q}p$ and $p\nmid k$, where $p$ is the characteristic of $\F_q$ (see Theorems \ref{thm:4.12} and \ref{thm:4.15}).
\end{itemize}\end{maintheorem}

\begin{maintheorem}\label{main:3}
Let $q=p^m$ be a prime power. Let $r$ be a positive integer.
Let  $I=\{0,1,2,\cdots,k-r-1, k-r+1,k-r+2,\cdots,k\}$ with $2\le r\le k-1$ and $3\le k\le n/2$. Then there exists a $q$-ary  non-RS MDS code $C(\bGa,I)$ which can be explicitly constructed, if one of the following conditions is satisfied.
%Then there exists a $q$-ary  non-RS MDS code $C(\bGa,I)$ which can be explicitly constructed, if one of the following conditions is satisfied.
\begin{itemize}
\item[{\rm (i)}] $2k\le n\le \sqrt[r]{r!\cdot q}/k$ for a prime $q$ (see Corollary \ref{cor:6.2});  or
\item[{\rm (ii)}] $2k\le n \le \max\{1, [\frac{\sqrt[r]{r!\cdot p}}{k}]\}\cdot p^{t}$, where $p$ is the characteristic of $\F_q$  and $t=\lfloor \frac{m-1}{r}\rfloor$ (see Theorems \ref{thm:6.3} and \ref{thm:6.4}).
\end{itemize}
\end{maintheorem}

\begin{rem}{Although some non-RS MDS codes have been constructed, the code lengths remain limited.
Compared with the results in \cite{Chen24}, the Main Theorem \ref{main:2} can largely improve that in \cite{Chen24}, especially for $r=1$ and a prime $p$.
 Furthermore, in comparison with the results  for $(n-1)|(q-1)$ and  $(n-1)|(q-1)$) in \cite{Beelen17} and  $(n-1)|q$ in \cite{LC24}, the Main Theorems \ref{main:2} and \ref{main:3} offer greater flexibility in choosing code lengths, thereby allowing us to cover a larger range of lengths.	
}

\end{rem}

\begin{rem}{In this paper, we consider $I=\{0,1,2,\cdots,k-r-1, k-r+1,k-r+2,\cdots,k\}$ and define
$C(\bGa,I)=\{(f(\Ga_1),f(\Ga_2),\dots,f(\Ga_n)):\; f\in V_I\}$. It is obvious to see that its Schur square $C^{\star 2}(\bGa,I)$ is an $[n,2k,\ge n-2k]$ code,
which has a large minimum distance. However, the Schur square of the Roth-Lempel code \cite{RL} has a significantly smaller minimum distance. Therefore, the codes we construct are not equivalent to the Roth-Lempel codes.}
\end{rem}

\subsection{Organization}
This paper is organized as follows. In section \ref{sec:2}, we introduce the definitions and properties of generalized Reed-Solomon codes, arithmetic progression sets and Schur products of linear codes.
Section \ref{sec:3} presents a general framework of non-RS MDS codes by selecting appropriate evaluation polynomials and points to ensure that all polynomials of degree at least $k$ have at most $k-1$ zeros in the set of evaluation points. Section \ref{sec:4} provide existence results of non-RS MDS codes using combinatorial counting arguments. In Sections \ref{sec:5.1} and \ref{sec:6}, explicit constructions of non-RS MDS codes are given by utilizing the general framework established in Section \ref{sec:3}.

\section{Preliminaries}\label{sec:2}

In this section, we review  key definitions and established results related to generalized Reed-Solomon codes, arithmetic progression sets and Schur products of linear codes, which will be useful for the  discussions in the following sections of this paper.

\subsection{Generalized Reed-Solomon codes}\label{sub:2.1}
Let $\F_q$ be the finite field of $q$ elements, where $q$ is a power of a prime $p$.
Let $\{\Ga_1,\Ga_2,\dots,\Ga_n\}$  be the set of $n$ distinct elements of $\F_q$ and denote by $\bGa=\{\alpha_1,\Ga_2,\cdots,\alpha_n\}$.
Let $V$  be a $k$-dimensional vector subspace over $\F_q$ of the polynomial ring $\F_q[x]$.
We define the evaluation map of $V$ on $\bGa$ by
\[ ev_{\bGa}(\cdot): V\rightarrow \F_q^n,\quad \quad f\mapsto (f(\alpha_1),f(\alpha_2),\cdots,f(\alpha_n)).\]
If we consider $V_{k-1}=\{f(x)\in \F_q[x]: \deg(f(x))\le k-1\}$, then it is well-known that the image $ev_{\bGa}(V_{k-1})$ is the $[n,k]$ Reed-Solomon code with evaluation set $\bGa$.
Since any nonzero $f\in V_{k-1}$ is of degree at most $k-1$, the polynomial $f$ has at most $k-1$ zeros. This shows that the $[n,k]$ RS code has minimum distance at least $n-k+1$, and hence it is an MDS code.

Two codes $C_1$ and $C_2$ and  are called equivalent if $C_2$ can be obtained from $C_1$ by a permutation of coordinates and a component-wise multiplication with some vector $\bv=(v_1,v_2,\cdots,v_n)\in \F_q^n$ of Hamming weight $n$, i.e., $v_i\neq 0$ for $i=1,2,\cdots,n$.
If two codes $C_1$ and $C_2$ are equivalent, then they share the same code parameters.

A generalized Reed-Solomon (GRS) code is in fact a code which is equivalent to an RS code. Choose $n$ nonzero elements $v_1,v_2,\cdots,v_n$ of $\F_q$ and put $\bv=(v_1,v_2,\cdots, v_n)$.
Then the GRS code can be defined as
\[GRS_k(\bGa,\bv):=\{(v_1f(\alpha_1),v_2f(\alpha_2),\cdots,v_nf(\alpha_n)):f(x)\in\F_q[x], \deg(f(x))\le k-1\}.\]
Thus, the above GRS code is also an $[n,k,n-k+1]$ MDS code and its dual code is an MDS code as well.
Otherwise, we say that a code is a non-RS MDS code if it is not equivalent to a RS code.

\subsection{Arithmetic progression sets}\label{sub:2.2}
For a finite subset $I=\{i_1,i_2,\dots,i_\ell\}\subseteq \mathbb{Z}$ of integers with $i_1<i_2<\cdots<i_\ell$, we say that $I$ is an arithmetic progression if $i_j-i_{j-1}=i_{j+1}-i_j$ for all $2\le j\le \ell-1$. Otherwise, we say that $A$ is a non-arithmetic progression. Note that a set of one or two elements can be defined to be an arithmetic progression.

An inverse problem in additive number theory is the problem in which we attempt to deduce properties of the set $I$ from properties of their sumset $I+I=\{a+b:\; a,b\in I\}$.
The following result is the simplest case in additive number theory \cite[Theorem 1.2]{Na}.
\begin{lemma}\label{lem:2.1}
Let $I$ be a set of $\ell$ integers. Then $|I+I|\ge 2\ell-1$.
Moreover, if $I$ is a set of  $\ell$ integers with $|I+I|=2\ell-1$, then $I$ is an arithmetic progression.
\end{lemma}

The above Lemma \ref{lem:2.1} tells us that $I+I$ has cardinality at least $2\ell$ for any non-arithmetic progression $I$. For any $\ell\ge 1$, there exists an arithmetic progression of $\ell$ elements. For instance, $I=\{0,1,2,\dots,\ell-1\}$. On the other hand, for any integer $\ell\ge 3$, there also exists a non-arithmetic progression of $\ell$ elements, for example  $I=\{0,1,2,3,\dots,\ell-2,\ell\}$ or $I=\{0,1,2,\dots,\ell-3,\ell-1,\ell\}$.

\subsection{Schur products of linear codes}\label{sub:2.3}
There has been a growing interest in the Schur products of linear codes, not only as an independent object of study, but also due to its relevance in various applications of coding theory, see \cite{Cramer15, Rand15} and references therein.

\begin{defn}[Schur product]
For  two vectors $\bx=(x_1,x_2,\cdots,x_n)$, $\by=(y_1,y_2,\cdots,y_n)$ in $\F_q^n$, the Schur product or component-wise product of $\bx$ and $\by$ is defined as $$\bx\star \by=(x_1y_1,x_2y_2,\cdots,x_ny_n).$$
For two linear codes $C_1, C_2\subseteq\F_q^n$, the Schur product of $C_1$ and $C_2$ is a linear subspace of $\F_q^n$ generated by all Schur products $\bc_1\star\bc_2$ with $\bc_1\in C_1$, $\bc_2\in C_2$, namely
\[C_1\star C_2:=\span_{\F_q}\{\bc_1\star\bc_2: \bc_1\in C_1,\bc_2\in C_2\}.\]
If $C_1=C_2=C$, then we call $C^{\star2}:=C\star C$ the Schur square of $C$.
\end{defn}

For any $[n,k]$-linear code $C$, it is easy to see that $\dim(C^{\star2})\le\min\{n,k(k-1)/2\}$.
If $C$ is an MDS code, then there is a lower bound given in \cite{Rand15} saying that $\dim(C^{\star2})\ge \min\{2k-1,n\}$.
There is a criterion to determine whether an MDS code is a GRS code from \cite[Theorem 23]{MZ}. However, we need the following simpler result.
\iffalse
 \begin{lemma}\label{lem:2.2}
Let $C\subseteq \F_q^n$ be an MDS code with $\dim(C)\le (n-1)/2$. Then the code $C$ is a Reed-Solomon code if and only if $\dim(C^{\star2})=2\dim(C)-1$.
\end{lemma}
\fi
\begin{lemma}\label{lem:2.2}
Let $C\subseteq \F_q^n$ be an MDS code with $\dim(C)\le n/2$.  If $\dim(C^{\star2})\ge 2\dim(C)$, then the code $C$ is not equivalent to a Reed-Solomon code.
\end{lemma}
\begin{proof}
Let $v_1,v_2,\cdots,v_n$ be nonzero elements of $\F_q$ and $\bv=(v_1,v_2,\cdots, v_n)$.
Suppose that the code $C$ was an $[n,k]$ MDS code given by the generalized Reed-Solomon code
\[GRS_k(\bGa,\bv)=\{(v_1f(\alpha_1),v_2f(\alpha_2),\cdots,v_nf(\alpha_n)):f(x)\in\F_q[x], \deg(f(x))\le k-1\}.\]
It is easy to verify that the Schur square of $C$ is
$$C^{\star 2}=GRS_{2k-1}(\bGa,\bv\star \bv).$$
Hence, the dimension of Schur square $C^{\star 2}$ is $2k-1$, which is a contradiction!
\end{proof}

For $k\le n/2$, the dimension of the Schur square of any $[n,k,n-k+1]_q$ generalized Reed-Solomon code is $2k-1$ from the proof of Lemma \ref{lem:2.2}.
Hence, any $[n,k,n-k+1]_q$ MDS code satisfying $\dim(C^{\star2})\ge 2k$ is a non-RS MDS code.
Thus, one can determine whether it is non-Reed-Solomon or not by computing the dimension of its Schur square.

\section{A general framework of non-Reed-Solomon MDS codes}\label{sec:3}
In this section, we provide a general framework for constructing non-Reed-Solomon (non-RS) MDS codes. A sufficient condition has been given to characterize whether a given MDS code is non-Reed-Solomon in Lemma \ref{lem:2.2}.

Reed-Solomon codes are constructed by using polynomials of degree at most $k-1$  for a code of dimension $k$.
In fact, if we choose different set $V$ of evaluation polynomials, the resulting code may not be equivalent to an RS code.
In this section, we focus on the construction of MDS codes which are not equivalent to RS codes by selecting $V$ to include
polynomials of degree at least $ k$.
However, we must carefully control the minimum distance of such codes to ensure that they are MDS codes.

Now let us describe the general framework of constructing our codes as follows.
Let $q$ be a prime power and $\F_q$ be the finite field with $q$ elements.
Let $\{\Ga_1,\Ga_2,\cdots,\Ga_n\}$ be a subset of $\F_q$.
Choose a vector space $V$ of dimension $k\le n/2$ in the polynomial ring $\F_q[x]$.
%such that $\gcd(f(x),\prod_{i=1}^n(x-\alpha_i))\le k-1$  for any polynomial $ f(x)\in V\setminus \{0\}$.
Define the code $C$ by
%\[C=\{(f(\alpha_1),\cdots,f(\alpha_n)): f(x)\in V, \gcd(f(x),\prod_{i=1}^n(x-\alpha_i))\le k-1 \mbox{ for } f(x)\in V\setminus \{0\}\}.\]
\begin{equation}\label{eq:1}
C=\{(f(\alpha_1),f(\alpha_2),\cdots,f(\alpha_n)): f(x)\in V\}.
\end{equation}
Then $C$ is a non-RS MDS code if the following conditions are satisfied:
\begin{itemize}
\item [(i)]  $\deg (\gcd(f(x),\prod_{i=1}^n(x-\alpha_i)))\le k-1$  for any polynomial $ f(x)\in V\setminus \{0\}$;
\item [(ii)]  $C$ is non-RS, i.e., $\dim(C^{\star2})\ge2k$.
\end{itemize}

It is easy to see that the length of $C$ is $n$ and the dimension of $C$ is $k$, i.e., $C$ is an $[n,k]$-linear code.
Furthermore, we can show that the code $C$ is an MDS code. It remains to prove that the minimum distance $d$ of $C$ achieves the Singleton bound $n-k+1$.
Since the degree of $\gcd(f(x),\prod_{i=1}^n(x-\alpha_i))$ is at most $k-1$ for any $f(x)\in V\setminus \{0\}$, the polynomial $f(x)$ has at most $k-1$ zeros in the set $\{\Ga_1,\Ga_2,\cdots,\Ga_n\}$, i.e., the Hamming weight of any nonzero codeword in $C$ is at least $n-k+1$. Thus, the minimum distance of $C$ is exactly $n-k+1$ from the Singleton bound, i.e., $C$ is an $[n,k,n-k+1]$ MDS code.

The main goal in the following is to choose suitable subspaces $V$ and evaluation sets $\{\Ga_1,\Ga_2,\cdots,\Ga_n\}$ such that the degree of $\gcd(f(x),\prod_{i=1}^n(x-\alpha_i))$ is at most $k-1$ for any polynomial $f(x)\in V\setminus \{0\}$.

We introduce the following notations.
\begin{itemize}
	\item $I=\{i_1,i_2,\cdots,i_k\}\subseteq \mathbb{N}$ with $|I|=k$;
	\item $m_I=\max\{i:i\in I\}$;
	\item  $V_I={\rm \span}_{\F_q}\{x^i:i\in I\}=\{\sum_{i\in I} a_ix^i: a_i\in \F_q \text{ for all } i\in I\}$;
	\item $Z(f)=\{\alpha\in \F_q: f(\alpha)=0\}$ for $f(x)\in V$;
	\item $S=\{ A\subseteq\F_q: |A|=k \mbox{ and } \exists f\in V_I\setminus \{0\} \mbox{ such that }  A\subseteq Z(f)\}$;
	\item  $T\subseteq \F_q$ with $|T|=n$;
	\item $T_k=\{A\subseteq T: |A|=k\}$ by all $k$-subsets of $T$.
	\end{itemize}

\begin{prop}\label{prop:3.1}
Let $n$ and $k$ be positive integers with $3\le k\le n/2$.
If there exists a subset $T$ of $F_q$ with size $n$ satisfying $T_k\cap S=\varnothing$, then the code $C(T,I)=\{(f(\alpha))_{\alpha\in T}: f\in V_I\}$ is an $[n,k]$ MDS code.
Moreover, if the set $I$ is a non-arithmetic progression, then $C(T,I)$ is not equivalent to a Reed-Solomon code.
\end{prop}
\begin{proof}
Since $T_k\cap S=\varnothing$,  any nonzero polynomial $f(x)$ in $V_I$ has at most $k-1$ zeros in the evaluation set $T$. It follows that the Hamming weight of codeword $(f(\alpha))_{\alpha\in T}$ is at least $n-k+1$.
From the Singleton bound, the minimum distance of $C(T,I)$ is $n-k+1$. Hence, $C(T,I)$ is an $[n,k]$ MDS code.

In order to show that the code $C(T,I)$ is non-Reed-Solomon, we need to determine the dimension of the Schur square of $C(T,I)$ given by $C^{\star2}(T,I)=\{(g(\alpha))_{\alpha\in T}: g\in \span_{\F_q}\{x^i: i\in I+I\}\}$.
If $I$ is a non-arithmetic progression, then the set $I+I=\{i+j:i,j\in I\}$ has cardinality at least $2k$ by Lemma \ref{lem:2.1}.
Hence, the Schur square $C^{\star2}(T,I)$ has dimension at least $2k$. By Lemma \ref{lem:2.2}, it is not equivalent to a Reed-Solomon code.
\end{proof}

\section{Existence results for non-Reed-Solomon MDS codes}\label{sec:4}
In this section, we use a combinatorial counting argument to show that there exists at least one set $T$ such that $T_k\cap S=\varnothing$ which ensures the MDS property. Furthermore, the resulting MDS code is non-RS as long as we choose an arbitrary non-arithmetic progression.

\begin{lemma}\label{lem:3.2}
Let $S$ and $m_I$ be defined as in Section \ref{sec:3}. Then the size of $S$ is upper bounded by $$|S|\le \frac{q^k-1}{q-1}\binom{m_I}{k}.$$
\end{lemma}
\begin{proof}
For any monic polynomial $f\in V_I\setminus \{0\}$, $f$ has at most $m_I$ roots since $\deg(f)\le m_I$.
Then the number of $k$-subsets of the zeros of $f$ in $\F_q$ is upper bounded by
$$| \{ A\subseteq\F_q: |A|=k,  A\subseteq Z(f)\}|\le \binom{m_I}{k}.$$
This lemma holds true, since there are $(q^k-1)/(q-1)$ pairwise distinct monic polynomials in $V_I\setminus \{0\}$.
\end{proof}

\begin{lemma}\label{lem:3.3}
If  $\binom{q}{n}>\frac{q^k-1}{q-1}\binom{m_I}{k}\binom{q-k}{n-k}$, then there exists a set $T\subseteq\F_q$ with cardinality $|T|=n$ such that $$T_k\cap S=\varnothing.$$
\end{lemma}
\begin{proof}
For every set $A\in S$, there are at most $\binom{q-k}{n-k}$ sets $T$ satisfying $|T|=n$ and $A\in T_k$. By Lemma \ref{lem:3.2}, there are at most $\frac{q^k-1}{q-1} \binom{m_I}{k}\binom{q-k}{n-k}$ sets $T\subseteq\F_q$ such that $|T|=n$ and $T_k\cap S\neq\varnothing$. As long as $\binom{q}{n}>\frac{q^k-1}{q-1} \binom{m_I}{k}\binom{q-k}{n-k}$, one can find such a required set $T$.
\end{proof}

\begin{theorem}\label{thm:3.4}
Let $q$ be a prime power.
Let $n$ and $k$ be positive integers with $3\le k\le n/2$.
Let $I\subseteq \mathbb{N}$ be a non-arithmetic progression with size $|I|=k$ and let $m_I=\max\{i:i\in I\}$.
If $\binom{q}{n}>\frac{q^k-1}{q-1}  \binom{m_I}{k}\binom{q-k}{n-k}$, then there exists an $[n, k]$
non-RS MDS code.
 \end{theorem}
\begin{proof}
This theorem follows from Lemma \ref{lem:3.3} and Proposition \ref{prop:3.1}.
\end{proof}

\begin{rem}\label{rem:3.5}{\rm
The inequality in Theorem \ref{thm:3.4} is equivalent to
$$n(n-1)\cdots (n-k+1)<\frac{q(q-1)\cdots(q-k+1) \cdot (q-1)}{(q^k-1)\cdot \binom{m_I}{k}}.$$
This implies that $n=O\left(\frac{q^{\frac1k}}{m_I}\right)$. Hence, there exists an $[n,k]$ non-RS MDS code   with length $n=O\left(\frac{q^{\frac1k}}{m_I}\right)$.
}
\end{rem}

Now let us choose a particular non-arithmetic progression set $I$.
Let $I$ be the set $\{0,1,2,\cdots,k-2,k\}$ with cardinality $k$ and $V_I$ be the vector space spanned by $\{x^i: i\in I\}$ over $\F_q$.

\begin{theorem}\label{thm:4.2}
Let $n$ and $k$ be positive integers with $3\le k\le n/2$.
If $\binom{q}{n}>\binom{q}{k-1} \binom{q-k}{n-k}$, then there exists an $[n, k]$ non-RS MDS code.
 \end{theorem}
\begin{proof}
If $A\in S$, then $A$ must be the set of $k$ pairwise distinct roots of some polynomial $f(x)=x^k+\sum_{i=0}^{k-2}c_ix^i$ in $V_I$.
By the Vieta's theorem, the sum of all roots of $f(x)$ is zero.
Let $\{\Ga_{i_1},\Ga_{i_2},\cdots,\Ga_{i_k}\}$ be the set of roots of $f(x)$. Thus, we have
$$\sum_{j=1}^k \Ga_{i_j}=\Ga_{i_1}+\Ga_{i_2}+\cdots+\Ga_{i_k}=0.$$
Trivially, the size of $S$ can be upper bounded by $$|S|\le \binom{q}{k-1}.$$
For every set $A\in S$, there are at most $\binom{q-k}{n-k}$ sets $T$ satisfying $|T|=n$ and $A\in T_k$.
Hence, there are at most $\binom{q}{k-1} \binom{q-k}{n-k}$ sets $T\subseteq\F_q$ such that $|T|=n$ and $T_k\cap S\neq\varnothing$.
If  $\binom{q}{n}>\binom{q}{k-1} \binom{q-k}{n-k}$, then there exists a set $T\subseteq\F_q$ with cardinality $|T|=n$ such that $$T_k\cap S=\varnothing.$$
Moreover, it is easy to see that $I=\{0,1,2,\cdots,k-2,k\}$ is a non-arithmetic progression set. Now this theorem follows from Proposition \ref{prop:3.1}.
\end{proof}

{\begin{rem}
By Theorem \ref{thm:4.2}, if $n(n-1)\cdots (n-k+1)<(k-1)! \cdot (q-k+1)$, then there exists an $[n, k]$ non-RS MDS code.
Hence, the length of the non-RS MDS codes can achieve $O(q^{\frac{1}{k}})$. Note that non-RS MDS codes can be constructed with length $n=O(q^{\frac{1}{k}})$ in \cite{Chen24} using algebraic curves.
\end{rem}}

\section{Explicit constructions of non-RS MDS codes}\label{sec:5.1}
In this section, we provide explicit constructions of non-RS MDS codes applying the framework in Section \ref{sec:3}.
%In particular, we convert the construction of non-RS MDS codes to a combinatorial problem which determines the maximal size of evaluation sets $T\subseteq \F_q$ such that all their $k$-subset sums are nonzero.
In particular, we focus on the case $I=\{0,1,2,\cdots,k-2,k\}$ and $V_I=\span_{\F_q}\{x^i: i\in I\}$.
In order to obtain explicit constructions of non-RS MDS codes, we need to elaborately choose the set of evaluation points $T=\{\Ga_1,\Ga_2,\cdots,\Ga_n\}$ such that any polynomial $f(x)\in V_I$ has at most $k-1$ zeros in the evaluation set $T$.

\iffalse
Assume $I=\{0,1,2,\cdots,k-2,k\}$ and $V_I$ is the vector space spanned by $\{x^i: i\in I\}$ over $\F_q$.
For a polynomial $f(x)\in V_I$, denote $Z(f)$ as the set of zeros of $f(x)$ in $\F_q$.
Let $S=\{ A\subseteq\F_q: |A|=k \mbox{ and } \exists f\in V_I\setminus \{0\} \mbox{ such that }  A\subseteq Z(f)\}$.
For any subset $T=\{\Ga_1,\Ga_2,\cdots,\Ga_n\}$ of $\F_q$ with $|T|=n\ge 2k+1$, denote by $T_k=\{A\subseteq T: |A|=k\}$.

In the following context, we need to elaborately choose the evaluation set $T=\{\Ga_1,\Ga_2,\cdots,\Ga_n\}$.
If $A\in S$, then $A$ must be the set of $k$ pairwise distinct roots of some polynomial $f(x)=x^k+\sum_{i=0}^{k-2}c_ix^i$ in $V_I$.
By the Vieta's theorem, the sum of all roots of $f(x)$ is zero.
Let $\{\Ga_{i_1},\Ga_{i_2},\cdots,\Ga_{i_k}\}$ be the set of roots of $f(x)$. Thus, we have
\begin{equation}
\sum_{j=1}^k \Ga_{i_j}=\Ga_{i_1}+\Ga_{i_2}+\cdots+\Ga_{i_k}=0.
\end{equation}
\fi

\begin{prop}\label{prop:4.3}
Let $n$ and $k$ be positive integers with $3\le k\le n/2$.
Let $I$ be the set $\{0,1,2,\cdots,k-2,k\}$ and $V_I$ be the vector space spanned by the set $\{x^i: i\in I\}$ over $\F_q$.
Let $T=\{\Ga_1,\Ga_2,\cdots, \Ga_n\}$ be a subset of $\F_q$ with cardinality $n$ such that all its $k$-subset sums of $T$ are nonzero, i.e.,
\begin{equation}\label{eq:3}
\sum_{j=1}^k \Ga_{i_j}\neq 0
\end{equation}
for all pairwise distinct elements $\Ga_{i_j}$ with $1\le i_1<i_2<\cdots<i_k\le n$.
Then the code $C(T,I)=\{(f(\alpha_1),f(\alpha_2),\cdots,f(\alpha_n)): f\in V_I\}$ is an $[n,k]$ non-RS MDS code.
\end{prop}
\begin{proof}
For any nonzero polynomial $f(x)\in V_I$, the Hamming weight of the codeword $(f(\alpha_1),f(\alpha_2),\cdots,f(\alpha_n))$ is at least $n-k$.
Suppose that $f(x)$ has $k$ pairwise distinct zeros $\{\Ga_{i_1},\Ga_{i_2},\cdots, \Ga_{i_k}\}$ in $T$.
By the Vieta's theorem, we obtain
$$\sum_{j=1}^k \Ga_{i_j}=0,$$
which contradicts to Equation \eqref{eq:3}.
Thus, the minimum distance of $C$ is at least $n-k+1$ which achieves the Singleton bound.
Hence, the code $C(T,I)$ is an $[n,k]$ non-RS MDS code from Lemma \ref{lem:2.1} and Proposition \ref{prop:3.1}.
\end{proof}

By Proposition \ref{prop:4.3}, the existence of the set $T$ with $T_k\cap S=\emptyset$ in Proposition \ref{prop:3.1} is converted to find a subset of $\F_q$ such that all its $k$-subset sums are nonzero. In the following,
our main goal is to find such a subset $T=\{\Ga_1,\Ga_2,\cdots,\Ga_n\}$ of $F_q$ satisfying
$ \sum_{j=1}^k \Ga_{i_j}\neq 0$
for all $k$-subsets $\{\Ga_{i_1}, \Ga_{i_2},\cdots, \Ga_{i_k}\}$ of $T$.

\begin{cor}\label{cor:4.4}
Let $p$ be an odd prime. Let $n$ and $k$ be positive integers with $3\le k\le n/2$.
If $2k\le n<\frac{p}{k}+\frac{k+1}{2}$, then there exists a $p$-ary $[n,k]$ non-RS MDS code.
\end{cor}
\begin{proof}
Let $p$ be an odd prime. Over the finite field $\F_p$, let $\Ga_i=i-1+p\mathbb{Z}\in \F_p$ for any $1\le i\le n$.
If $n<\frac{p}{k}+\frac{k+1}{2}$, then it is easy to verify that
$$0< \sum_{j=1}^k \Ga_{i_j}\leq (n-1)+(n-2)+\cdots+(n-k)=kn-\frac{k(k+1)}{2}<p$$
for any $1\le i_1<i_2<\cdots<i_k\le n$.
By Proposition \ref{prop:4.3}, the code $C(T,I)=\{(f(\Ga_1),f(\Ga_2),\cdots,f(\Ga_n)): f\in V_I\}$ is a $p$-ary $[n,k]$ non-RS MDS code.
\end{proof}

{\begin{rem}
In \cite{Chen24}, non-RS MDS codes are constructed from algebraic curves which can achieve $O(q^{1/k})$. Corollary \ref{cor:4.4} shows that our codes have much longer length than that given in \cite{Chen24} when $q$ is a prime. Moreover, comparing with the results in \cite{Beelen17} and \cite{LC24} where the length  $(n-1)|(q-1)$ or  $(n-1)|q$,
or $n|(q-1)$, we can provide more flexibility in the choice of code length, allowing us to cover a larger range of lengths.
\end{rem}}

\subsection{Non-RS MDS codes via weight numerators}
In this subsection, we provide an explicit construction of non-RS MDS codes via weight numerators of linear codes.
Let $A_k$ be the number of codewords of $C$ with Hamming weight $k$ and $\sum_{i=0}^n A_ix^i$ be the weight numerator of $C$.

\begin{prop}\label{prop:4.5}
Let $q$ be a prime power.
Let $n$ and $k$ be positive integers with $3\le k\le n/2$.
Let $A_k$ be the number of codewords of $C$ with Hamming weight $k$.
If there exists a $q$-ary $[n,n-r,d\ge 3]$-linear code $C$ with $A_k=0$, then there exists a $q^r$-ary $[n,k]$ non-RS MDS code.
\end{prop}
\begin{proof}
Let $H=(h_1,h_2,\cdots,h_n)\in \F_q^{r\times n}$ be a parity-check matrix of $C$.
%If $d\ge 3$, then any $d-1\ge 2$ columns of $H$ are linearly independent from \cite[Theorem 4.5.6]{LX04}.
If there exist $k$ different columns such that $$h_{i_1}+h_{i_2}+\cdots+h_{i_k}=0,$$ then there exists a codeword of $C$ with Hamming weight $k$.
Thus, we derive a contradiction with the fact $A_k=0$.

Let $\{\beta_1,\beta_2,\cdots,\beta_r\}$ be a basis of $\F_{q^r}$ over $\F_q$.
Then  any $r$-dimensional column vector  $h_j=(a_{1,j},a_{2,j},\cdots,a_{r,j})^T$ over $\F_q$  can be identified with an element $\Ga_j=\sum_{i=1}^r a_{i,j} \beta_i\in \F_{q^r}$.
If $d\ge 3$, then any $d-1\ge 2$ columns of $H$ are linearly independent from \cite[Theorem 4.5.6]{LX04}.
Hence, columns $h_j$ can be viewed as pairwise distinct elements $\Ga_j$ in $\F_{q^r}$ satisfying
$$\sum_{j=1}^k \Ga_{i_j}\neq 0$$
for all pairwise distinct elements $\Ga_{i_j}$ with $1\le i_1<i_2<\cdots<i_k\le n$.
Now this proposition follows from \ref{prop:4.3}.
\end{proof}

\begin{cor}\label{cor:4.6}
Let $q$ be a prime power. If there exists a $q$-ary $[n,n-r,k+1]$-linear code, then there exists a $q^r$-ary $[n,k]$ non-RS MDS code.
\end{cor}
\begin{proof}
Let $C$ be a $q$-ary $[n,n-r,k+1]$-linear code. It is easy to see that the number of codewords of $C$ with Hamming weight $k$ is zero.
By Proposition \ref{prop:4.5}, there exists a $q^r$-ary $[n,k]$ non-RS MDS code.
\end{proof}

\iffalse
\begin{ex}
From the theory of BCH codes  \cite[Proposition 8.1.12]{LX04} and propagation rules \cite[Theorem 6.1.1]{LX04}, there exists a $q$-ary $[q^m-1, q^m-1-mk, k+1]$ BCH code.
By Corollary \ref{cor:4.6}, there exists a $q^{mk}$-ary $[q^{m}-1,k]$ non-RS MDS code. Such a non-RS MDS code has length $n=q^m-1=O(q^m)=O((q^{mk})^{1/k})$.
\end{ex}
\fi

\begin{ex}
From the theory of BCH codes  \cite[Proposition 8.1.14]{LX04} and propagation rules \cite[Theorem 6.1.1]{LX04}, there exists a $2^m$-ary $[2^m-1, 2^m-1-mt, 2t+1]$ BCH code. By Corollary \ref{cor:4.6}, there exists a $2^{mt}$-ary $[n=2^{m}-1,k=2t]$ non-RS MDS code. Such a non-RS MDS code has length $n=2^m-1=O(2^m)=O((2^{mt})^{2/k})$.
\end{ex}

Now we give some numerical examples using the known codes given in the code tables (see https://www.codetables.de).
\begin{ex}
\begin{itemize}
\item[(1)] From the code tables, there exists a $[256,240,5]_2$-linear code. For $q=2^{16}$, there exists a $q$-ary $[n=q^{1/2},k=4]$ non-RS MDS code by Corollary \ref{cor:4.6}.
\item[(2)] From the code tables, there exists a $[256,232,7]_2$-linear code. For $q=2^{24}$, there exists a $q$-ary $[n=q^{1/3},k=6]$ non-RS MDS code by Corollary \ref{cor:4.6}.
\item[(3)] From the code tables, there exists a $[243,235,4]_3$-linear code. For $q=3^{8}$, there exists a $q$-ary $[n=q^{5/8},k=3]$  non-RS MDS code by Corollary \ref{cor:4.6}.
\item[(4)] From the code tables, there exist $[243,223,7]_3$-linear code. For $q=3^{20}$, there exists a $q$-ary $[n=q^{1/4},k=6]$  non-RS MDS code by Corollary \ref{cor:4.6}.
\end{itemize}
\end{ex}

\begin{ex}\label{ex:4.9}
By \cite[Proposition 5.3.10]{LX04}, the extended binary Hamming code $\overline{Ham(r,2)}$ is a binary $[2^r,2^r-1-r,4]$-linear code.
 By Corollary \ref{cor:4.6}, there exists a $q$-ary $[q/2,3]$ non-RS MDS code for $q=2^{r+1}$.
%However, the previous result shows that $n=O(q^{1/3})$.
\end{ex}

\begin{ex}\label{ex:4.10}
%The Hamming code $Ham(r,q)$ is a binary $[\frac{q^r-1}{q-1},\frac{q^r-1}{q-1}-r,3]$-linear code.
The $q$-ary extended Hamming code  $\overline{Ham(r,q)}$  is a $q$-ary $[\frac{q^r-1}{q-1}+1,\frac{q^r-1}{q-1}-r,4]$-linear code.
By Corollary \ref{cor:4.6}, there exists a $q^{r+1}$-ary $[\frac{q^r-1}{q-1}+1,3]$ non-RS MDS code. Such a non-RS MDS code has length $n=\frac{q^r-1}{q-1}+1\ge (q^{r+1})^{\frac{r-1}{r+1}}$ for $r\ge 3$.
%However, the previous result shows that $n=O(q^{(r+1)/3})$.
\end{ex}

The above Example \ref{ex:4.10} shows that when $k=3$, we can construct long non-RS codes with length $n$ approaching the alphabet size of the code for odd $q$ and sufficiently large $r$.

\iffalse
\begin{ex}\label{ex:4.6}
Let $t\ge 3$ be a positive integer.
The extended Hamming code $\overline{Ham(t,2)}$  is a binary $[2^t,2^t-1-t,4]$-linear code.
However, the Hamming weights of its codewords are even.
For $q=2^{t+1}$ and any odd integer $3\le k\le (q-2)/4$, there exists a $q$-ary $[q/2,k]$ non-RS MDS code by Proposition \ref{prop:4.5}.
\end{ex}
\fi

\begin{cor}\label{cor:4.11}
Let $r\ge 3$ be a positive integer and $q=2^{r+1}$. For any odd integer $3\le k\le q/4$, there exists a $q$-ary $[q/2,k]$ non-RS MDS code.
\end{cor}
\begin{proof}
The extended Hamming code $\overline{Ham(r,2)}$  is a binary $[2^r,2^r-1-r,4]$-linear code.
Moreover, the Hamming weights of its codewords are even.
By Proposition \ref{prop:4.5}, there exists a $q$-ary $[q/2,k]$ non-RS MDS code.
\end{proof}

\begin{rem} From the above examples, we can see that non-RS MDS codes can be constructed explicitly, with lengths that improve upon the existence results presented in Section \ref{sec:4}. Moreover, the length of our non-RS MDS codes can achieve $q/2$ for $q=2^{r+1}$ and odd integer $3\le k\le q/4$.
\end{rem}

\subsection{Non-RS MDS codes via vector representations}
Motivated by the above construction, here we provide another construction of non-RS MDS codes by choosing the evaluation set $\{\Ga_1,\Ga_2,\cdots, \Ga_n\}$ deliberately such that $\sum_{j=1}^k \Ga_{i_j}\neq 0$.

Let $q=p^m$ be a power of $p$. Let $\F_q$ be the finite field with $q$ elements.
Let $\zeta$ be a generator of the extension field $\F_q=\F_p(\zeta)$ over $\F_p$.
Then the minimum polynomial of $\zeta$ over $\F_p$ has degree $m$, and
$\{1, \zeta,\zeta^2,\cdots, \zeta^{m-1}\}$ is a basis of $\F_q$ over $\F_p$.
The Equation \eqref{eq:3} in Proposition \ref{prop:4.3} can be satisfied if we
choose $\Ga_i=1+\sum_{j=1}^{m-1} a_{i,j} \zeta^j$ for any positive integer $1\le i\le n$.

\begin{theorem}\label{thm:4.12}
Let $p$ be a prime and $\F_q=\F_p(\zeta)$ be the finite field with $q=p^m$ elements.
Let $I$ be the set $\{0,1,2,\cdots,k-2,k\}$ and $V_I$ be the vector space spanned by $\{x^i: i\in I\}$ over $\F_q$.
Let $n$ and $k$ be positive integers with $n\le p^{m-1}$, $p\nmid k$ and $3\le k\le n/2$.
Let $T=\{\Ga_1,\Ga_2,\cdots, \Ga_n\}$ be a subset of $\F_q$ with  $\Ga_i=1+\sum_{i=1}^{m-1} a_{i,j} \zeta^j$  for any $1\le i\le n$.
Then the code $C(T,I)=\{(f(\alpha_1),f(\alpha_2),\cdots,f(\alpha_n)): f\in V_I\}$ is a $q$-ary $[n,k]$ non-RS MDS code.
\end{theorem}
\begin{proof}
From the choices of $\Ga_i$ and $p\nmid k$, it is easy to verify that
$$\sum_{j=1}^k \Ga_{i_j}=k+\sum_{j=1}^{m-1} c_j\zeta^j\neq 0$$
with some $c_j\in \F_p$ for any $1\le i_1<i_2<\cdots<i_k\le n$.
By Proposition \ref{prop:4.3}, the code $C(T,I)$ is a $q$-ary $[n,k]$ non-RS MDS code.
\end{proof}

In \cite{RL}, Roth and Lempel provided a construction of non-RS MDS codes where they transform the  construction of a $q$-ary $[n+2,k+1]$ non-RS MDS code  into a combinatorial problem that determines the maximum value $R(k,q)$ of the set
	$\{n: \exists \text{ distinct } \Gb_1,\Gb_2,\cdots,\Gb_n\in \F_q \text{ and } \delta\in \F_q \text{ such that } \sum_{j=1}^k \Gb_{i_j}\neq \delta \text{ for } 1\le i_1<i_2<\cdots<i_k\le n\}.$

 By Proposition \ref{prop:4.3}, our method of constructing a $q$-ary $[n,k]$ non-RS MDS code is converted to determining the maximum value $T(k,q)$ of the set
 $$\left\{n: \exists \text{ distinct } \Ga_1,\Ga_2,\cdots,\Ga_n\in \F_q \text{ s.t. } \sum_{j=1}^k \Ga_{i_j}\neq 0 \text{ for } 1\le i_1<i_2<\cdots<i_k\le n\right\}.$$
Since $0\in \F_q$, it is trivial that $T(k,q)\le R(k,q)$. On the other hand, we can show $T(k,q)\ge R(k,q)$ under the assumption $p\nmid k$.

\begin{lemma}\label{lem:4.13}
If  $p\nmid k$, then we have $T(k,q)=R(k,q)$.
\end{lemma}
\begin{proof}
It remains to prove that $T(k,q)\ge R(k,q)$ for $p\nmid k$.
Let $n=R(k,q)$. Then there exist pairwise distinct elements $\Gb_1,\Gb_2,\cdots,\Gb_n\in \F_q$ and some element $\delta\in \F_q$ such that
$$ \sum_{j=1}^k \Gb_{i_j}\neq \delta $$
for any $1\le i_1<i_2<\cdots<i_k\le n$.
If $p\nmid k$, then we let $\Ga_i=\Gb_i-\delta/k$ for any $1\le i\le n$. It is easy to verify that
$$  \sum_{j=1}^k \Ga_{i_j}=\sum_{j=1}^k (\Gb_{i_j}-\delta/k)=\sum_{j=1}^k \Gb_{i_j}-k\cdot (\delta/k)=\sum_{j=1}^k \Gb_{i_j} -\delta\neq 0 $$
for any $1\le i_1<i_2<\cdots<i_k\le n$. Hence, we have $T(k,q)\ge n=R(k,q)$.
\end{proof}

For $q=2^m$ with $m\ge 2$, we have
$R(3,q)=R(q/2-2,q)=R(q/2-1,q)=q/2+1$,
$R(k,q)=q/2$ for $4\le k\le q/2-3$ and $R(k,q)=k+2$ for $q/2-1\le k\le q-2$ by \cite[Lemma 4, Lemma 5 and Lemma 6]{RL}. Hence, we can obtain the following results by Lemma \ref{lem:4.13} which are similar to Example \ref{ex:4.9} and Corollary \ref{cor:4.11}.

\begin{theorem}\label{thm:4.14}
Let $q=2^m$ be a prime power. There exists a $q$-ary $[q/2+1,3]$ non-RS MDS code.
Moreover, for any odd integer $5\le k\le q/4$, there exists a $q$-ary $[q/2,k]$ non-RS MDS code.
\end{theorem}

Let $q$ be a power of an odd prime $p$. By \cite[Lemma 8]{RL}, we have $R(k,q)\ge k+2$ for $3\le k\le (q-3)/2$. The exact value of $R(k,q)$ remains open for odd $q$.

\begin{theorem}\label{thm:4.15}
%Let $q=p^m$ be a power of an odd prime $p$.
Let $p$ be an odd prime, $q=p^m$ and $\F_q=\F_p(\zeta)$ be the finite field with $q$ elements.
Let $k$ be a positive integer with $3\le k\le p-1$, let $I$ be the set $\{0,1,2,\cdots,k-2,k\}$ and $V_I$ be the vector space spanned by the set $\{x^i: i\in I\}$ over $\F_q$.
Let $n$ be a positive integer with $2k\le n\le [\frac{p}{k}]\cdot p^{m-1}+(p-k\cdot [\frac{p}{k}]-1)$.
Let $T=\{\Ga_1,\Ga_2,\cdots, \Ga_n\}$ be a subset of $\F_q$ with  $\Ga_i=a_{i,0}+\sum_{i=1}^{m-1} a_{i,j} \zeta^j$ with $1\le a_{i,0}\le [\frac{p}{k}]$ for any $1\le i\le [\frac{p}{k}]\cdot p^{m-1}$ and $a_{i,0}=[\frac{p}{k}]+1$ for any $[\frac{p}{k}]\cdot p^{m-1}+1\le i\le n$.
Then the code $C(T,I)=\{(f(\alpha_1),f(\alpha_2),\cdots,f(\alpha_n)): f\in V_I\}$ is a  $q$-ary  $[n,k]$ non-RS MDS code.
\end{theorem}
\begin{proof}
From the choices of $\Ga_i$, it is easy to verify that
$$\sum_{j=1}^k \Ga_{i_j}=\sum_{j=1}^k a_{i_j,0}+\sum_{j=1}^{m-1}c_j\zeta^j\neq 0$$
with some $c_j\in \F_p$ for arbitrary pairwise distinct elements $\Ga_{i_j}$ in $T$.
By Proposition \ref{prop:4.3}, the code $C(T,I)$ is a $q$-ary $[n,k]$ non-RS MDS code.
\end{proof}

\section{A generalization of explicit non-RS MDS codes}\label{sec:6}
In Section \ref{sec:5.1}, we provide explicit constructions of non-RS MDS codes for $I=\{0,1,2,\cdots,k-2,k\}$.
In this section, we give a generalization of the explicit construction of non-RS MDS codes provided in Section \ref{sec:5.1}
for general  $I=\{0,1,2,\cdots,k-r-1, k-r+1,k-r+2,\cdots,k\}=\{0,1,2,\cdots,k\}\setminus \{k-r\}$ with $k\ge 3$ and $1\le r\le k-1$.
In particular, if $r=1$, then the set $I$ is exactly the one given in Section \ref{sec:5.1}.

\begin{prop}\label{prop:6.1}
Let $T=\{\Ga_1,\Ga_2,\cdots, \Ga_n\}$ be a subset of $\F_q$ with cardinality $n\ge 2k$ such that
\begin{equation}\label{eq:4}
\sum_{1\le j_1<\cdots<j_r\le k} \Ga_{i_{j_1}}\Ga_{i_{j_2}}\cdots \Ga_{i_{j_r}} \neq 0
\end{equation}
for all pairwise distinct elements $\Ga_{i_j}$ with $1\le i_1<i_2<\cdots<i_k\le n$.
Then the code $C(T,I)=\{(f(\alpha_1),f(\alpha_2),\cdots,f(\alpha_n)): f\in V_I\}$ is an $[n,k]$ non-RS MDS code.
\end{prop}
\begin{proof}
For any nonzero polynomial $f(x)\in V_I$, the Hamming weight of the codeword $(f(\alpha_1),f(\alpha_2),\cdots,f(\alpha_n))$ is at least $n-k$.
Suppose that $f(x)$ has $k$ pairwise distinct zeros $\{\Ga_{i_1},\Ga_{i_2},\cdots, \Ga_{i_k}\}$ in $T$.
By the Vieta's theorem, we obtain
$$\sum_{1\le j_1<\cdots<j_r\le k} \Ga_{i_{j_1}}\Ga_{i_{j_2}}\cdots \Ga_{i_{j_r}}=0,$$
which contradicts to Equation \eqref{eq:4}.
Thus, the minimum distance of $C(T,I)$ is at least $n-k+1$ which achieves the Singleton bound.
Hence, the code $C(T,I)$ is an $[n,k]$ non-RS MDS code from Lemma \ref{lem:2.1} and Proposition \ref{prop:3.1}.
\end{proof}

In the following, our main goal is to explicitly construct a subset $T=\{\Ga_1,\Ga_2,\cdots,\Ga_n\}$ of $F_q$ satisfying
$ \sum_{1\le j_1<\cdots<j_r\le k} \Ga_{i_{j_1}}\Ga_{i_{j_2}}\cdots \Ga_{i_{j_r}}\neq 0$
for all $k$-subsets $\{\Ga_{i_1}, \Ga_{i_2},\cdots, \Ga_{i_k}\}$ of $T$.

\begin{cor}\label{cor:6.2}
Let $p$ be an odd prime. Let $n$ and $k$ be positive integers with $3\le k\le n/2$.
If $n\le \sqrt[r]{r!\cdot p}/k$, then there exists a $p$-ary $[n,k]$ non-RS MDS code.
\end{cor}
\begin{proof}
Let $p$ be an odd prime. Over the finite field $\F_p$, let $\Ga_i=i-1+p\mathbb{Z}\in \F_p$ for any $1\le i\le n$.
If $n\le \sqrt[r]{r!\cdot p}/k$, then it is easy to verify that
$$0< \sum_{1\le j_1<\cdots<j_r\le k} \Ga_{i_{j_1}}\Ga_{i_{j_2}}\cdots \Ga_{i_{j_r}}\leq \binom{k}{r} (n-1)(n-2)\cdots (n-r)< n^rk^r/r!\le p.$$
By Proposition \ref{prop:6.1}, there exists a $p$-ary $[n,k]$ non-RS MDS code.
\end{proof}

\begin{theorem}\label{thm:6.3}
Let $p$ be a prime and $\F_q=\F_p(\zeta)$ be the finite field with $q=p^m$ elements.
Let $k$ and $r$ be positive integers with $k\ge 3$, $1\le r\le k-1$ and $p\nmid \binom{k}{r}$.
Let $I$ be the set $\{0,1,2,\cdots,k-r-1, k-r+1,k-r+2,\cdots,k\}$ and $V_I$ be the vector space spanned by $\{x^i: i\in I\}$ over $\F_q$.
Let  $t=\lfloor \frac{m-1}{r}\rfloor$ and $n$ be a positive integer with $2k\le n\le p^{t}$.
%Let $k$ be a positive integer with $3\le k\le n/2$ and $p\nmid \binom{k}{r}$.
Let $T=\{\Ga_1,\Ga_2,\cdots, \Ga_n\}$ be a subset of $\F_q$ with  $\Ga_i=1+\sum_{i=1}^{t} a_{i,j} \zeta^j$  for any $1\le i\le p^t$.
Then the code $C(T,I)=\{(f(\alpha_1),f(\alpha_2),\cdots,f(\alpha_n)): f\in V_I\}$ is a $q$-ary $[n,k]$ non-RS MDS code.
 \end{theorem}
 \begin{proof}
  Let $\Ga_i=1+\sum_{j=1}^{t} a_{i,j} \zeta^j$ be elements of $\F_q$ for $1\le i\le p^t$.
 For any positive integer $k$ satisfying $\binom{k}{r}\not\equiv 0(\text{mod } p)$ and $1\le i_1<i_2<\cdots < i_k\le n$, we obtain
 $$\sum_{1\le j_1<\cdots<j_r\le k} \Ga_{i_{j_1}}\Ga_{i_{j_2}}\cdots \Ga_{i_{j_r}}=\binom{k}{r}+\sum_{j=1}^{rt} c_j \zeta^j\neq 0 $$
for some $c_j\in \F_p$ with $1\le j\le rt$.
By Proposition \ref{prop:6.1}, the code $C(T,I)$ is a $q$-ary $[n, k]$ non-RS MDS code.
 \end{proof}

Combining Corollary \ref{cor:6.2} with Theorem \ref{thm:6.3}, we have the following result.

\begin{theorem}\label{thm:6.4}
%Let $q=p^m$ be a power of an odd prime $p$.
Let $p$ be an odd prime, $q=p^m$ and $\F_q=\F_p(\zeta)$ be the finite field with $q$ elements.
Let $k$ and $r$ be positive integers with $3\le k\le \sqrt[r]{r!\cdot p}$ and $1\le r\le k-1$.
%Let $r$ be a positive integer.
Let $I$ be the set $\{0,1,2,\cdots,k-r-1, k-r+1,k-r+2,\cdots,k\}$ and $V_I$ be the vector space spanned by $\{x^i: i\in I\}$ over $\F_q$.
Let $t=\lfloor \frac{m-1}{r}\rfloor$ and let $n$ be a positive integer with $2k\le n\le [\frac{\sqrt[r]{r!\cdot p}}{k}]\cdot p^{t}$.
Let $T=\{\Ga_1,\Ga_2,\cdots, \Ga_n\}$ be a subset of $\F_q$ with  $\Ga_i=a_{i,0}+\sum_{i=1}^{t} a_{i,j} \zeta^j$ with $1\le a_{i,0}\le [\frac{\sqrt[r]{r!\cdot p}}{k}]$ for any $1\le i\le n$.
Then the code $C(T,I)=\{(f(\alpha_1),f(\alpha_2),\cdots,f(\alpha_n)): f\in V_I\}$ is a $q$-ary $[n,k]$ non-RS MDS code.
\end{theorem}
\begin{proof}
Since $1\le a_{i,0}\le [\frac{\sqrt[r]{r!\cdot p}}{k}]$ for any $1\le i\le n$, it is easy to see that
 $$1\le \sum_{1\le j_1<\cdots<j_r\le k} a_{i_{j_1,0}}a_{i_{j_2,0}}\cdots a_{i_{j_r,0}}<\binom{k}{r}\frac{r!\cdot p}{k^r}\le p.$$
Hence, we have
 $$\sum_{1\le j_1<\cdots<j_r\le k} \Ga_{i_{j_1}}\Ga_{i_{j_2}}\cdots \Ga_{i_{j_r}}=\sum_{1\le j_1<\cdots<j_r\le k} a_{i_{j_1,0}}a_{i_{j_2,0}}\cdots a_{i_{j_r,0}}+\sum_{j=1}^{rt} c_j \zeta^j\neq 0 $$
for arbitrary pairwise distinct elements $\Ga_{i_j}$ in $T$.
By Proposition \ref{prop:6.1}, the code $C(T,I)$ is a $q$-ary $[n,k]$ non-RS MDS code.
\end{proof}

\begin{rem}{In fact, we can generalize our framework to a more general case using the idea of twist RS codes.
Let $V=span_{\F_q}\{1,x,\cdots,x^{k-r-1},x^{k-r+1},\cdots,x^{k-1},(-1)^r\delta x^{k-r}+x^k\}$ for any $1\le r\le k-1$. Then each monic polynomial $f\in V$ has the form $x^k+(-1)^r\delta x^{k-r}+\sum_{i=0}^{k-1}a_i x^{i}$. In this case, the corresponding code $C=\{(f(\alpha_1),f(\alpha_2),\cdots,f(\alpha_n)): f\in V\}$ is an $[n,k]$ non-RS MDS code if $\sum_{1\le j_1<\cdots<j_r\le k} \Ga_{i_{j_1}}\Ga_{i_{j_2}}\cdots \Ga_{i_{j_r}} \neq \delta$  for all $k$-subsets of $\{\Ga_1,\Ga_2,\cdots,\Ga_n\}$. Therefore, for $r=1$, the code $C=\{(f(\alpha_1),f(\alpha_2),\cdots,f(\alpha_n)): f\in V\}$ is an $[n,k]$ non-RS MDS code if $\sum_{1\le i_1<\cdots<i_r\le k} \Ga_{i_{j}} \neq \delta$ for all $k$-subsets of $\{\Ga_1,\Ga_2,\cdots,\Ga_n\}$.
However, we are only interested in the  case when $\delta=0$ in this paper. }\end{rem}

\section{Conclusions}

The main goal of this paper is to construct new families of non-RS MDS codes. We begin by presenting a general framework for constructing non-RS type MDS codes. This approach involves selecting an appropriate set of evaluation polynomials and a corresponding set of evaluation points such that all nonzero polynomials have at most
$k-1$ zeros within the evaluation set. Additionally, both the existence and explicit constructions of these codes are provided. Constructing non-RS MDS codes for all possible lengths remains an interesting and challenging problem. However, using our framework, we believe it is possible to obtain more new families of non-RS MDS codes.

\end{document}